%% file: main.tex
\documentclass[a4paper,onecolumn,11pt,unpublished,accepted=2025-09-23]{quantumarticle}
\pdfoutput=1

\usepackage[utf8]{inputenc}
\usepackage[english]{babel}

\usepackage{changepage}
\usepackage{csquotes}
\usepackage[backend=biber, style=numeric, sorting=none,url=false,isbn=false,doi = false,date =year,eprint=true,]{biblatex}
\newbibmacro{string+doi}[1]{
  \iffieldundef{doi}{#1}{\href{https://doi.org/\thefield{doi}}{#1}}}
\DeclareFieldFormat{title}{\usebibmacro{string+doi}{#1}}
\DeclareFieldFormat[article,thesis]{title}{\usebibmacro{string+doi}{\mkbibquote{#1}}}
\DeclareFieldFormat{journaltitle}{#1\isdot}

\usepackage{microtype}
\usepackage[T1]{fontenc}
\usepackage{amsmath}
\usepackage{amssymb}
\usepackage{amsthm}
\usepackage{dsfont} %identity operator
\usepackage{mathtools}
\usepackage[colorlinks=true, linkcolor=blue, urlcolor=blue, citecolor = blue]{hyperref}
\usepackage{scalerel}
\usepackage{thmtools,thm-restate}
\usepackage{tikz}
\usepackage{lipsum}
\usepackage{soul}
\usepackage{mdframed}
\usepackage{tikz-cd}
\usepackage{quiver}

\newtheorem{theorem}{Theorem}[section]
\newtheorem*{theorem*}{Theorem}
\newtheorem{lemma}[theorem]{Lemma}
\newtheorem{corollary}[theorem]{Corollary}

\theoremstyle{definition}
\newtheorem{definition}{Definition}[section]
\theoremstyle{remark}
\newtheorem*{remark}{Remark}
\newtheorem{conjecture}{Conjecture}

\newcommand{\CL}{\mathcal{L}}
\newcommand{\CLp}{\mathcal{L}^{\perp}}
\newcommand{\CS}{\mathcal{S}}
\newcommand{\CC}{\mathcal{C}}

\newcommand{\CH}{\mathcal{H}}
\newcommand{\CT}{\mathcal{T}}

\newcommand{\CD}{\mathcal{D}}
\newcommand{\R}{\mathbb{R}}
\newcommand{\C}{\mathbb{C}}
\newcommand{\N}{\mathbb{N}}
\newcommand{\T}{\mathbb{T}}

\newcommand{\mgkp}{\mathcal{M}_{\textnormal{GKP}}}
\newcommand{\mgkpd}{\mathcal{M}_{\textnormal{GKP}}^{\CD}}
\newcommand{\mlat}{\mathcal{M}_{\textnormal{lat}}}
\newcommand{\GU}{\overline{U}_{G}}

\definecolor{light-gray}{gray}{0.65}
\newmdtheoremenv[linewidth=2 pt, topline=false, bottomline=false, linecolor= light-gray, rightline=false, 
leftmargin=1pt, innerleftmargin=0.4em, rightmargin=0pt, innerrightmargin=0pt, innertopmargin=-5pt ,
innerbottommargin=1pt, splittopskip=\topskip, splitbottomskip=0.3\topskip, 
skipabove=0.6\topsep]
{example}{Example}

\newcommand{\Sp}{\textnormal{Sp}}
\newcommand{\ket}[1]{| #1 \rangle}
\setstcolor{blue}

\newcommand{\sqr}{\CD}

\begin{document}

\title{Fiber Bundle Fault Tolerance of GKP Codes}

\author{Ansgar G. Burchards}
\affiliation{Dahlem Center for Complex Quantum Systems, Physics Department, Freie
Universit{\"a}t Berlin, Arnimallee 14, 14195 Berlin, Germany}
\author{Steven T.\ Flammia}
\affiliation{Department of Computer Science, Virginia Tech, Alexandria, USA}
\affiliation{Phasecraft Inc., Washington DC, USA}
\author{Jonathan Conrad}
\affiliation{Dahlem Center for Complex Quantum Systems, Physics Department, Freie
Universit{\"a}t Berlin, Arnimallee 14, 14195 Berlin, Germany}
\affiliation{Helmholtz-Zentrum Berlin f{\"u}r Materialien und Energie, Hahn-Meitner-Platz 1, 14109
Berlin, Germany}

\maketitle

\begin{abstract}
  We investigate multi-mode GKP (Gottesman--Kitaev--Preskill) quantum error-correcting codes from a geometric perspective. 
  First, we construct their moduli space as a quotient of groups and exhibit it as a fiber bundle over the moduli space of symplectically integral lattices. 
  We then establish the Gottesman--Zhang conjecture for logical GKP Clifford operations, showing that all such gates arise from parallel transport with respect to a flat connection on this space. 
  Specifically, non-trivial Clifford operations correspond to topologically non-contractible paths on the space of GKP codes, while logical identity operations correspond to contractible paths.
\end{abstract}

\section{Introduction}
\label{sec:introduction}
\input{section_introduction}

\section{Preliminaries}
\label{sec:prelims}
\input{section_background}

\section{Symplectic automorphism groups and moduli space of lattices}
\label{sec:lattices}
\input{section_symplectic_automorphisms.tex}

\section{GKP codes}
\label{sec:GKP}
\input{section_GKP.tex}

\section{The moduli space of GKP codes}
\label{sec:GKP_moduli_space}
\input{section_GKP_moduli_space}

\section{Geometric fault tolerance}
\label{sec:geometric_fault_tolerance}
\input{section_geometric_FT.tex}

\section{Conclusion}
\label{sec:conclusion}

We have investigated the class of GKP quantum error correcting codes from a geometric point of view. 
In particular, we have determined the moduli space of all multi-mode GKP codes and formally established it as a bundle over the moduli space of symplectically integral lattices. 
We have shown that this bundle is non-trivial and labeled its connected components by a unique type, each corresponding to a canonical decomposition of the encoded logical system into qudits. 
Furthermore, we have shown that the well-known Gaussian implementations of logical Pauli and Clifford operations on GKP codes induce a flat projective connection on the GKP moduli space, proving for these operations the fault tolerance conjecture of Gottesman and Zhang~\cite{gottesman_fibre_2017}. 
Our work constitutes the first proof of the conjecture for a set of continuous variable codes.

We believe that our classification of GKP codes will be a useful tool aiding in the practical implementation of multi-mode codes and hope that this work can stimulate additional interest in the relationship between quantum error correction and geometry. 
Future work could investigate different types of error-correcting codes within the same framework. 
Of particular interest here would be to understand the geometry of the space of \textit{concatenated} GKP codes, as this class of codes provides an embedding of natively qubit- or qudit stabilizer codes into multi-mode GKP codes \cite{Conrad_2022_lattice_perspective} and the respective GKP moduli spaces are expected to encode specific characteristics, such as the existence of transversal gates.
Another important avenue for future exploration is extending the geometric perspective to fault-tolerant non-Clifford and non-unitary operations as well as bosonic codes beyond the GKP code~\cite{Jain_2024_quantum, denys_2024_quantum}. A detailed investigation of the formalism in the context of approximate, finite-energy GKP codes would also be of interest.

\section*{Acknowledgments} 
We thank V. Albert, J. Eisert, F. Arzani, J.
Magdalena de la Fuente and Alex Townsend-Teague for many insightful and helpful discussions. AB gratefully acknowledges V. Albert for hosting him during a research stay at the University of Maryland, which contributed to the progress of this project. AB and JC also acknowledge the support from the BMBF (RealistiQ, MUNIQC-Atoms, PhoQuant, QPIC-1, and QSolid), the DFG (CRC 183, project B04, on entangled states of matter),
the Munich Quantum Valley (K-8), the ERC (DebuQC), Quantum Berlin, as well as the Einstein Research Unit on quantum devices.

\appendix

\section{Vector and fiber bundles}
\label{app:fiber_intro}
\input{section_FiberBundleIntroduction}

\section{A unique normal form for the Gram matrix}
\label{app:gram_unique_normal_form}
\input{section_Frobenius_form}

\printbibliography

\end{document}

%% file: section_introduction.tex
Quantum computers hold the promise of providing a speedup over classical computers on specific computational tasks such as factoring \cite{Shor_1997_polynomial}.
In order for quantum computers to realize this potential 
and perform arbitrarily large quantum computations, the key requirement, besides universality, is that of fault tolerance \cite{Shor_1997_polynomial, aharonov_1999_faulttolerant, Knill_1998_resilient}. 
This means that given high enough control accuracy and sufficiently weak noise with sufficiently low correlations, arbitrarily long and
precise computations are possible with only a modest cost in additional overhead. 
On a macroscopic level, it means that sufficiently good individual components are robust enough to suppress errors arbitrarily well when combined in the right way and in sufficient numbers.

While a widely used, general, and intuitive concept, the term fault tolerance is often used in an ad-hoc fashion tailored to specific details of the context or platform under discussion. 
Three examples include \emph{transversal gates}~\cite{gottesman_1997_stabilizer}, which are fault-tolerant as they do not propagate errors; \emph{topological gates}~\cite{Kitaev_2003_fault_tolerant}, implemented through the braiding of localized anyonic excitations on the ground space of a topological error-correcting code; and \emph{Gaussian unitary operations}, which perform logical Clifford operations on GKP codes~\cite{Gottesman_2001_encoding, Weedbrook_2012_gaussian}, with fault tolerance due to the fact that the relative increase of displacement errors is bounded.
These different types of protocols span a variety of physical systems and conceptual frameworks, and while all three types of gates are sensibly deemed fault tolerant within their respective domains, they do not seem to arise from a single universal construction method upon first inspection. 
For example, while the braiding of anyons is clearly topological in nature, this does not seem to be the case for either transversal or Gaussian unitary operations. 

In their work~\cite{gottesman_fibre_2017}, Gottesman and Zhang conjectured that in fact all types of fault-tolerant gates can be regarded as topological, so that a universal perspective on fault tolerance becomes possible. 
Here we prove the Gottesman--Zhang conjecture for logical Clifford gates on multi-mode GKP codes by showing that they arise as the parallel transport maps induced by a flat projective connection on a bundle over the space of all GKP codes.
Denoting by $\mgkp$ the manifold of all GKP stabilizer groups on $n$ modes and by $U_{G}$ the set of Gaussian unitary operators, we prove the following theorem.
\begin{theorem} 
\label{thm:geometric_ft}
Any path $\CC: [0,1] \rightarrow \mgkp$ has a unique Gaussian unitary implementation $U: [0,1] \rightarrow U_{G}/\{e^{i \phi} \mathds{1}\}$ with $U(0) = \mathds{1}$ and $\CC(t) = U(t) \CC(0) U^{\dagger}(t)$. 
Letting Gaussian unitaries act on the trivial bundle over $\mgkp$ with standard fiber the space $\R^{2n}$ of displacement operators as $U\big(\CC, D(\xi)\big) = (U \CC U^{\dagger}, U D(\xi) U^{\dagger})$, then the associated connection is flat. 
\end{theorem}
Intuitively, implementing a path on a space of error correcting codes via some unitary $U(t)$ guarantees protection from errors at all times during gate implementation. Flatness of the connection generating parallel transport both guarantees robustness to control inaccuracies and that implemented logical gates are a feature of the path's topology only.
The proof of Theorem~\ref{thm:geometric_ft} can be found in Section~\ref{sec:geometric_fault_tolerance}.
To the best of our knowledge this is the first proof of the Gottesman--Zhang conjecture for a multi-mode bosonic code; see also Ref.~\cite{conrad_2024_gkp_rosetta} by the authors which provides a different perspective on the single-mode case. 

The paper is organized as follows. Section~\ref{sec:prelims} reviews the necessary background on continuous variable quantum systems and lattice theory. In Section~\ref{sec:lattices}, we explore the symmetry properties of symplectically integral lattices and derive an expression for their moduli space. Section~\ref{sec:GKP} reviews some background on GKP codes, formulated in lattice theoretic language. Section~\ref{sec:GKP_moduli_space} presents a construction of the GKP moduli space, while Section~\ref{sec:geometric_fault_tolerance} offers a proof of the Gottesman--Zhang conjecture for GKP Clifford operations. We conclude in Section~\ref{sec:conclusion}.

%% file: section_background.tex
In this section we introduce some preliminary notions regarding continuous variable quantum systems, the theory of lattices, and quantum error correction. We also give a short introduction to vector and fiber bundles in Appendix~\ref{app:fiber_intro}.
While we provide all necessary background, we refer the interested reader to Refs.~\cite{Weedbrook_2012_gaussian, Gerry_Knight_2004_optics} for additional background on continuous variable physics, to Ref.~\cite{ConwaySloane_1988} for a broad overview on lattice theory and to Refs.~\cite{baez_1994_Gauge, frankel_2011_Geometry} for more detailed treatments of bundle theory.

\subsection{Continuous variable quantum systems}
An operator basis on the Hilbert space of $n$ modes is given by the well known \emph{position and momentum operators}, denoted $\hat{q}$ and $\hat{p}$, on each mode. 
Throughout we will reference these operators through the ‘qqpp’-ordered vector
\begin{equation}
    \hat{x} = (\hat{q}_1, \dots, \hat{q}_n, \hat{p}_1, \dots , \hat{p}_n )^{T} \, .
\end{equation}
The position and momentum operators satisfy the commutation relations
\begin{equation}
    [\hat{x}_{i} , \hat{x}_j] = i J_{ij}
\end{equation}
where the matrix
\begin{equation}
    J = J_{2n} = \begin{pmatrix}
        0 & \mathds{1}_{n \times n} \\
        -\mathds{1}_{n \times n} & 0 
    \end{pmatrix} 
\end{equation}
denotes the standard \emph{symplectic form}. 
For brevity we will refer to the standard symplectic vector space $(\R^{2n}, J_{2n})$ simply as $\R^{2n}.$
A second set of important operators are the \emph{displacement operators}, which are generated by linear functions in the position and momentum operators
\begin{equation}
    D(\xi) \coloneqq \text{exp}\big(-i \sqrt{2 \pi} \xi^{T} J \hat{x}\big)\, .
\end{equation}
These operators owe their name to their action on the position and momentum operators
\begin{equation}
    D(\xi)^\dagger \hat{x} D({\xi}) = \hat{x} + \sqrt{2 \pi} \xi ,
\end{equation}
and satisfy the multiplication and commutation relations
\begin{equation}
\begin{split}
    D(\xi) D(\eta)  & =  e^{-i \pi  \xi^{T} J \eta} D(\xi + \eta) \\  & =  e^{-i 2 \pi  \xi^{T} J \eta} D(\eta) D(\xi) .
\end{split}
\end{equation}
\subsection{Gaussian unitaries}
Any unitary operator that is generated by a degree-2 polynomial in the position and momentum operators is called a \emph{Gaussian unitary}. 
We denote the set of Gaussian unitaries by $U_G$ and the Gaussian unitaries modulo phases by $\GU \coloneqq U_{G}/\{ e^{i \phi} \mathds{1} \} $. 
The full set $\GU$ is conveniently parametrized as $U_{\eta, S} \coloneqq D(\eta)\hat{S}$, where $\hat{S}$ denotes the \emph{metaplectic representation} of $S \in \text{Sp}_{2n}(\R).$

\subsection{Metaplectic representation}
The metaplectic representation is a double valued unitary representation of the symplectic group $\Sp_{2n}(\R)$, up to factors of $\pm 1$. 
Formally this is a representation of the double cover of $\Sp_{2n}(\R)$ whose values on the two branches differ by a factor of $-1$ \cite{guaita_2024_representation}. We denote the image of $S \in \Sp_{2n}(\R)$ under the metaplectic representation by $\hat{S}$. All such $\hat{S}$ are Gaussian unitaries and the representation satisfies the intertwining relation
\begin{equation}
    \hat{S}^\dagger \hat{x} \hat{S} = S \hat{x} \, .
\end{equation}
From the definition of displacement operators one further obtains
\begin{equation}
    \hat{S} D(\xi) \hat{S}^\dagger = D(S\xi).
\end{equation}

\subsection{Lattices}
Throughout we will make extensive use of the concept of a \emph{lattice}.
\begin{definition}[Lattice]
    A lattice $\CL$ is a discrete, additive subgroup of a vector space. 
\end{definition}
We will exclusively consider lattices contained within some symplectic vector space $V$ and define the rank of a lattice as the dimension of the space spanned by its constituent vectors. A lattice is called full rank if $\text{rank}(\CL) = \text{dim}(V)$.
\begin{definition}[Dual Lattice]
    The (symplectic) dual lattice of a full-rank lattice $\CL$ is 
    \begin{equation}
        \CLp \coloneqq \Big\{ \eta \in V \, \big\vert\, \forall \xi \in \CL: \eta^{T} J \xi \in \mathbb{Z}  \Big\},
    \end{equation}
    where $V$ is the symplectic vector space containing $\CL$.
\end{definition}
A lattice is called (symplectically) \emph{self-dual} if $\CL = \CLp$ and  
 integral if $\CL \subseteq \CLp$.
We refer to any set of linearly independent vectors $\{\xi_{i} \}$ whose integer linear combinations span $\CL$ as a basis of $\CL$ and to the matrix $M = (\xi_1, \dots, \xi_{2n})^{T}$ as a \emph{generator} of $\CL$. 
Any two generators of the same lattice $M, M'$ are related via left multiplication by a unimodular matrix, i.e.\ $M = U M'$ with $U$ an invertible integer matrix. 
If the lattice is full rank then $M$ is invertible and the dual lattice has a canonical \emph{dual generator} and basis $M^{\perp} = (\xi_{1}^{\perp}, \dots, \xi_{2n}^{\perp})^{T}$ defined via $M^{\perp} J  M^{T} = \mathds{1}$.

Given a choice of generator one defines the associated \emph{Gram matrix} as the matrix whose entries are the symplectic inner products between basis vectors.
\begin{definition}[Gram Matrix]
    The Gram matrix associated to a lattice generator $M$ is
     \begin{equation}
         A = MJM^{T} \, ,
     \end{equation}
      explicitly its entries are given in terms of basis elements as $A_{ij} = \xi_{i}^{T} J \xi_{j}$.
\end{definition}
For integral lattices we have $A \in \mathbb{Z}^{2n \times 2n}$ and we will refer to lattices with $A \in 2\mathbb{Z}^{2n \times 2n}$ as \emph{symplectically even}.

It can been shown that the Gram matrix of every lattice can be brought into \emph{standard form} via an appropriate choice of generator \cite{Conrad_2022_lattice_perspective}.
\begin{definition}
\label{def:standard_form}
    A Gram matrix $A$ is said to be in standard form if $A = J_{2} \otimes D$
     where $D$ is a diagonal matrix with positive integer entries, which are non-increasing along the diagonal.
\end{definition}

This standard form was already introduced in the original paper~\cite{Gottesman_2001_encoding} and classifies lattices up to symplectic equivalence \cite{Conrad_2022_lattice_perspective}. 
A drawback of this standard form, however, is that it is not unique for any given symplectically integral lattice. 
In order to remedy this we introduce a unique version of the standard form in Section~\ref{sec:lattices}.

\subsection{Error-correcting codes}
A quantum error correcting code is a subspace of a Hilbert space satisfying the Knill--Laflamme conditions, which pertain to the protection of encoded information from noise \cite{Knill_2000_theory}. 
A common and useful way of defining such codes is through a set of stabilizing operators.
\begin{definition}[Stabilizer Code]
    Given a set of operators $\{O_1, \dots, O_n \}$ we call the joint eigenspace
    \begin{equation}
        \CC = \Big\{ \ket \psi \in \mathcal{H} \, \big\vert\, \forall i\colon O_i \ket \psi = \ket \psi  \Big\}
    \end{equation}
    the \emph{stabilizer code} associated to the set of operators. 
    The group $\CS = \langle O_1, \dots, O_{n} \rangle$ is called a \emph{stabilizer} of the code.
\end{definition}
In most treatments of qubit stabilizer codes the operators $O_{i}$ are restricted to be drawn from the Pauli group, leading to a class of highly structured codes. 
Analogously, GKP codes are stabilizer codes stabilized by displacement operators. 
As we are interested in idealized GKP codes, whose member states are formally of infinite energy, we choose to define GKP codes indirectly through their stabilizer groups.
\begin{restatable}{definition}{GKPDefinition}
    A GKP code is specified by a stabilizer of the form 
    \begin{equation}
    \label{eq:gkp_stabilizer}
        \CS = \Big\langle e^{i\phi_1} D(\xi_1), \dots, e^{i \phi_{2n}} D(\xi_{2n}) \Big\rangle 
    \end{equation}
    where the set $\CL = \big\{ \xi \in \R^{2n} \mid \exists \phi\colon e^{i \phi} D(\xi ) \in \CS \big\}$ forms a full-rank, symplectically integral lattice.
\end{restatable}

The phase components $\phi_{i}$ are commonly suppressed in the literature since important error correction properties, such as distance and encoded logical dimension of the resulting code, are independent of them. 
However, since we aim to fully parametrize the set of all existing GKP codes it is necessary to explicitly include them here.
The requirement for the set $\CL$ to form a lattice is a direct consequence of the fact that the stabilizer forms a group and integrality follows from the requirement that its elements must be mutually commuting.
Demanding the lattice $\CL$ to be full rank is equivalent to encoding only a finite dimensional quantum system as opposed to a full mode. The close connection between GKP codes and lattices has already been discussed in the original work~\cite{Gottesman_2001_encoding} as well as more recently in Refs.~\cite{Conrad_2022_lattice_perspective, Royer_2022_encoding, Harrington_Thesis, Harrington_2001_achievable}.
We now turn to discuss some of the theory of lattices before returning to GKP codes in Section~\ref{sec:GKP}.

%% file: section_symplectic_automorphisms.tex
Given a lattice $\mathcal{L}$, any symplectic transformation on the embedding vector space results in another lattice $S\mathcal{L}$. 
Two lattices related by such a transformation are said to be symplectically equivalent. 
The set of such transformations preserving a given lattice is known as its associated \textit{symplectic automorphism group}, given by
\begin{equation}
    \text{Aut}^{S}(\mathcal{L}) \coloneqq \Big\{ S \in \text{Sp}_{2n}(\mathbb{R}) \, \big\vert\, S\mathcal{L} = \mathcal{L} \Big\}.
\end{equation}
The automorphism group is directly related to logical Clifford operations on GKP codes, as we will discuss in Section~$\ref{subsec:gkp_cliffords}$.
We are here interested in the structure of the moduli space $\text{Sp}_{2n}(\mathbb{R}) /\text{Aut}^{S}(\mathcal{L})$ of lattices symplectically equivalent to a given lattice $\CL$. 
We begin by determining the structure of the automorphism group for an arbitrary full rank lattice.

\subsection{Structure of symplectic automorphism groups}
Let $M$ be a generator of a full rank lattice $\CL$. 
Under a symplectic operation $S$ the generator matrix transforms to $M S^{T}$ and thus yields the original lattice exactly if this too is a generator of $\CL$. 
As any two generators are related by left multiplication with a unimodular matrix we obtain the following characterization
\begin{equation}
    \text{Aut}^{S}(\mathcal{L}) = \Big\{ S \in \text{Sp}_{2n}(\mathbb{R}) \, \big\vert\, \exists U \in GL_{2n}( \mathbb{Z}): \,  UM  = M S^{T} \Big\}\, .
\end{equation}
In general any unimodular matrix $U$ transforms any basis of $\CL$ into another basis by left-action on the generator matrix $M \rightarrow U M$. 
However, not all such basis transformations can arise from a symplectic operation.
To see this, we introduce the \textit{adjoint action}, $\text{Ad}_{X}(Y) = X Y X^{-1}$. 
Then using the fact that the generator matrix $M$ is invertible and the relation 
\begin{equation}
    \label{eq: UMeqMadU}
    U M = M \text{Ad}_{M^{-1}}(U)
\end{equation}
we see that the basis change specified by $U$ is implementable by a symplectic operation if and only if $\text{Ad}_{M^{-1}}(U)$ is symplectic. 
From this reasoning, we obtain the following characterization of symplectic automorphisms as a function of the Gram matrix.
\begin{lemma}\label{lemma:symplectic_implementable_basis_trafo}
    The basis transformation $U \in \textnormal{GL}_{2n}( \mathbb{Z})$, with respect to generator $M$, can be symplectically implemented if and only if the map $\bullet \rightarrow U \bullet U^{T}$ preserves the Gram matrix $A = M J M^{T}.$
\end{lemma}
\begin{proof}
    Suppose $U$ is symplectically implementable, then $\text{Ad}_{M^{-1}}(U)$ is symplectic and thus $U A U^{T} = U M J (UM)^{T} = M \text{Ad}_{M^{-1}}(U) J \text{Ad}_{M^{-1}}(U)^{T} M^{T} = M J M^{T} = A.$ Here we have used the fact that symplectic matrices preserve the symplectic form $J$. For the converse, suppose now that $A = U A U^{T}$. Then $\text{Ad}_{M^{-1}}(U) J \text{Ad}_{M^{-1}}(U)^{T} = M^{-1} U M J M^{T} U^{T} M^{-T} = M^{-1} U A U^{T} M^{-T} = M^{-1} A M^{-T} = J.$ Thus $\text{Ad}_{M^{-1}}(U)$ is symplectic.
\end{proof}

Considering the fact that every Gram matrix can be brought into standard form we arrive at the following simpler expression for the automorphism group:
\begin{corollary}
\label{cor:automorphism_symplectic_form}
    The symplectic automorphism group $\textnormal{Aut}^{\textnormal{S}}(\mathcal{L})$ of a lattice with standard form Gram matrix $A = J_{2} \otimes D$  is isomorphic to the generalized integer symplectic group
    \begin{equation}
        \textnormal{Sp}_{2n}(\mathbb{Z}; D) \coloneqq \Big\{  U \in \textnormal{SL}_{2n}(\mathbb{Z}) \, \big\vert\, U  (J_{2} \otimes D) U^{T} = J_{2} \otimes D \Big\}.
    \end{equation}
\end{corollary}
\begin{proof}
This follows by independence of the automorphism group from the chosen generator and Lemma~\ref{lemma:symplectic_implementable_basis_trafo}. 
The restriction to $\textnormal{SL}_{2n}(\mathbb{Z})$ is a consequence of the fact that the determinant of a symplectic matrix must equal $1$.
\end{proof}
As the elements of $\textnormal{Sp}_{2n}(\mathbb{Z}; D)$ are formally not symplectic but instead unimodular matrices, we refer to this group as the \textit{unimodular representation} of the symplectic automorphism group~\cite{Birkenhake_2004_complex}. 
Given the basis $M$ bringing the Gram matrix into standard form it holds that $\text{Aut}^{S}(\mathcal{L}) = \text{Ad}_{M^{-1}}(\text{Sp}_{2n}(\mathbb{Z}; D))$ element-wise.

An interesting and much investigated class of GKP codes is that of \emph{scaled codes}. These are GKP codes related to scaled self-dual lattices, i.e. those with $\CL = \sqrt{\lambda} \CL_{0}$ where $\lambda \in \N$ and $\CL_{0}$ self-dual \cite{Gottesman_2001_encoding, Conrad_2022_lattice_perspective, Harrington_2001_achievable, Harrington_Thesis}.
Noticing that in the case $D \propto \mathds{1}$ one has $\text{Sp}_{2n}(\mathbb{Z}; D) = \text{Sp}_{2n}(\mathbb{Z})$ we find that the automorphism groups of scaled self-dual lattices in particular reduce to the standard symplectic group over the integers. 
\begin{corollary}
    Let $\mathcal{L}$ be a scaled self-dual lattice, then 
    \begin{equation*}
        \textnormal{Aut}^{\textnormal{S}}(\mathcal{L}) \cong \textnormal{Sp}_{2n}(\mathbb{Z}).
    \end{equation*}
\end{corollary}
\begin{proof}
    Since $\mathcal{L}$ is a scaled lattice it is of the form $\mathcal{L} = \sqrt{\lambda} \mathcal{L}_{0}$ where $\lambda \in \mathbb{N}$ and $\mathcal{L}_{0}$ is symplectically self-dual i.e.\ $\mathcal{L}_{0}^\perp = \mathcal{L}_{0}$. 
    The standard form of the Gram matrix of symplectically self-dual lattices is $A = M_{0} J M_{0}^{T}= J_{2} \otimes \mathds{1}$. 
    As the generator matrix $M$ of $\mathcal{L}$ is obtained as $M = \sqrt{\lambda} M_{0}$ we find that $A = M J M^{T} = \lambda J_{2} \otimes \mathds{1} = \lambda J_{2n}.$ The automorphism group is thus $\textnormal{Sp}_{2n}(\mathbb{Z}, \lambda \mathds{1}) = \textnormal{Sp}_{2n}(\mathbb{Z})$.
\end{proof}

\subsection{The moduli space of symplectically integral lattices}

In order to determine the moduli space of symplectically integral lattices we make use of the following corollary slightly adapted from Ref.~\cite{Conrad_2022_lattice_perspective}:

\begin{corollary}\label{cor:generator_standard_form}
    Let $M \in \R^{2n \times 2n}$ describe the basis of a full-rank symplectically integral lattice with Gram matrix $A = MJM^{T}.$ Without loss of generality assume $M$ is chosen such that $A = J_{2}\otimes D$ is in standard form. The lattice specified by the generator 
    \begin{equation*}
        N_{D} = \bigoplus_{j=1}^{n} \sqrt{D_{jj}}\mathds{1}_{2\times 2}
    \end{equation*}
    is symplectically equivalent to the one specified by $M$. 
\end{corollary}
We will refer to the lattice generated by $N_{D}$ as the \emph{standard lattice} associated to the Gram matrix. 
The following example shows that the standard form of a Gram matrix introduced in Def.~\ref{def:standard_form} is not unique. 
\begin{example}
\label{example:unique_standard_form}
To see that the standard form is not unique consider e.g. the standard form Gram matrix $A = J_2 \otimes \text{diag}(5, 2)$ and the unimodular matrix
\begin{equation}
    U = \begin{pmatrix}
    2 & 0 & 0& 5\\
    1 & 0 & 0 & 2\\
    0 & -5 & -4 & 0\\
    0 & 1 & 1 & 0\\
    \end{pmatrix},
\end{equation}
which satisfy $UAU^{T} = J_2 \otimes \text{diag}(10, 1)$. 
Physically these two standard forms correspond to distinct decompositions of the logical space into either a 5-dimensional qudit and a qubit or a single 10-dimensional qudit as we discuss in  Section~\ref{subsec:gkp_paulis}.
\end{example}
Fortunately, further restricting the standard form allows one to find a unique form of the Gram matrix, which we show in Appendix~\ref{app:gram_unique_normal_form}. 
\begin{restatable}{corollary}{FrobeniusStandardForm}
    \label{corr:Frobenius_standard_form}
    (Frobenius standard form)
    Every symplectically integral lattice has a generator bringing the Gram matrix into standard form 
    \begin{equation}
        A = J_2 \otimes \CD
    \end{equation} 
    such that the matrix elements of $\CD$ satisfy $d_n |d_{n-1}| \dots |d_1$. Moreover this normal form is unique.
\end{restatable}
We say a Gram matrix with the property $d_{n}|d_{n-1}|\dots|d_{1}$ (read $d_{n}$ divides $d_{n-1}$ etc.) is in Frobenius standard form. The transformed Gram matrix $A = J_{2} \otimes \text{diag}(10, 1)$ of Example~\ref{example:unique_standard_form} is hence in standard form.
We henceforth refer to the Frobenius standard form value of $D$ as the \textit{type} of a lattice and denote it by $\CD$ in order to distinguish it from the non unique form.

Equipped with the above results we may determine the moduli space of symplectically integral lattices.
\begin{theorem}
\label{thm:sympl_integral_lattice_moduli_space}
    The moduli space of full rank symplectically integral lattices in $\big(\R^{2n}, J \big)$ is
    \begin{equation*}
        \mlat \coloneqq \bigcup_{\CD}\textnormal{Sp}_{2n}(\mathbb{R})/\textnormal{Sp}_{2n}(\mathbb{Z}; \CD)
    \end{equation*}
     where the index $\CD$ runs over all positive integer diagonal matrices with non-increasing elements along the diagonal satisfying $d_{n}|d_{n-1}|\dots|d_{1}.$
\end{theorem}
\begin{proof}
    By Corollary~\ref{corr:Frobenius_standard_form} each lattice $\CL$ has a unique type $\CD$ satisfying the specifications in the theorem. 
    Given the standard lattice $\CL_{\sqr}$ generated by $N_{\CD}$ as defined in Corollary~\ref{cor:generator_standard_form} there exists $S \in \Sp_{2n}(\R)$ with $S\CL_{\sqr} = \CL$ by the same corollary. 
    One also has that $S \CL_{\sqr} = S' \CL_{\sqr}$ if and only if $S = S'K$ with $K \in \text{Aut}^{S}(\CL_{\sqr})$ and it follows that we may write the moduli space of type-$\CD$ lattices as the quotient $\textnormal{Sp}_{2n}(\mathbb{R})/\text{Aut}^{S}(\CL_{\sqr})$. Employing Corollary~\ref{cor:automorphism_symplectic_form} we may denote this space $\textnormal{Sp}_{2n}(\mathbb{R})/\textnormal{Sp}_{2n}(\mathbb{Z}; \CD)$ where $U \in \textnormal{Sp}_{2n}(\mathbb{Z}; \CD)$ embeds into $\textnormal{Sp}_{2n}(\mathbb{R})$ via $ U \mapsto N_{\sqr}^{-1} U N_{\sqr}$. 
    As a consequence of Corollary~\ref{cor:generator_standard_form} a lattice of each type exists and therefore the moduli space of symplectically integral lattices is the union of these spaces. 
\end{proof}

\begin{remark}
The space $\mlat$ can also be considered from the perspective of abelian varieties; see Ref.~\cite{Birkenhake_2004_complex} for more details. 
\end{remark}

We will see below that non-identity Clifford gates on GKP codes arise from non-contractible loops in $\mlat$. While the space $\mlat$ is in general high-dimensional,
in the single-mode case its connected components are the $3$-dimensional quotient spaces $\text{Sp}_{2}(\R) / \text{Sp}_{2}(\mathbb{Z})$ which are topologically equivalent to a 3-sphere with a trefoil knot $\gamma$ removed \cite{milnor_introduction_2016}. We illustrate this space in Fig.~\ref{fig:moduli_space} and refer the interested reader to references~\cite{milnor_introduction_2016, ghys_2006_Lorenz} for more details on the topology of this space. 
Clifford gates on single-mode GKP codes can then be seen to arise from non-contractible loops that interlink nontrivially with the trefoil knot. We also note that the first homotopy group of $S^{3} - \gamma$ is given by $B_3$, the braid group on three strands. The construction hence also yields a homomorphism from this braid group to the Clifford group modulo Paulis.

%% file: section_GKP.tex
We now return to the topic of GKP codes.
Their connection to the theory of lattices originates from the form of their stabilizer groups which we recall here.

\GKPDefinition*

Given a generating set as in Eq.~\eqref{eq:gkp_stabilizer}, we refer to the vector of phases $\{\phi_i\}_{i=1}^{2n}$ as the \emph{sector} of the code. Crucially the sector of a given stabilizer depends on the chosen lattice generator $M = (\xi_{1}, \dots, \xi_{2n})^{T}$ and distinct choices of pairs ($M$, $\phi$) can yield the same stabilizer.
We refer to the choice $\phi=0$ as the trivial sector and note that this choice specifies a unique stabilizer independent of the chosen generator for every even lattice. 
In Section~\ref{sec:GKP_moduli_space} we will be occupied with the study of the set of all stabilizers of this form.

First let us note a striking feature of GKP codes --- they allow for the implementation of both logical Pauli and Clifford operations by Hamiltonians that are, respectively, degree-1 or degree-2 polynomials in the position and momentum operators.

\subsection{Pauli operators}
\label{subsec:gkp_paulis}
Given a lattice basis such that the Gram matrix assumes standard form $A = J_2 \otimes D$, the Gram matrix of the dual lattice is given by $A^{\perp} = J_{2} \otimes D^{-1}$. 
Consider the displacement operators corresponding to the canonical dual lattice basis vectors
\begin{equation}
    \overline{X}_{i} = D(\xi_{i}^{\perp}), \quad \overline{Z}_{i} = D(\xi_{i+n}^{\perp})  
\end{equation}
for $i \in \{1, \dots, n \}$.
These are logical operators as they commute with all stabilizers and commute pairwise except for the nontrivial commutation relations
\begin{equation}
    \overline{Z}_{i}\overline{X}_{i} = \text{exp}(2 \pi i / d_{i} )\overline{X}_{i} \overline{Z}_{i} \, .
\end{equation}
When restricting to the code space one further has
\begin{equation}
    \overline{Z}_{i}^{d_i} = \overline{X}_{i}^{d_i} = \mathds{1}
\end{equation}
up to phases.
It is apparent that the set of operators generated by $\{\overline{Z}_{i}, \overline{X}_{j}\}$ forms a representation of a multi-qudit Heisenberg-Weyl group with, in general, mixed local dimensions, and naturally decomposes the code space into a product of tensor factors of dimensions $\{d_{i}\}$. 
The fact that the standard form of a symplectically integral lattice's Gram matrix is not unique is hence equivalent to the availability of multiple such decompositions.

\subsection{Clifford operators}
\label{subsec:gkp_cliffords}
Logical Clifford operations must preserve the logical Pauli group under adjoint action. 
As the Gaussian unitary matrices preserve the set of displacement operators it is natural to construct Clifford operations from this set. 
Let us consider the trivial sector of a GKP stabilizer of type $\CD$
\begin{equation}
    \CS_{\sqr} \coloneqq \Big\{ e^{i \phi(\xi)} D(\xi) \, \big\vert\, \xi \in \CL \Big\}.
\end{equation}
Here $\phi(\xi)$ denote phases that occur even in the trivial sector as a consequence of the fact that products of commuting displacements may yield the negative of a displacement if the respective symplectic product is odd. 
See Example~\ref{example:square_GKP_qutrit_Sgate} for such a case. 
A Gaussian unitary operator $U_{\eta, S}$ implements a logical Clifford operation exactly if it belongs to $N(\CS_{\sqr})$ --- the normalizer of $\CS_{\sqr}$ within the Gaussian unitaries
\begin{equation}
    N(\CS_{\sqr}) = \Big\{U_{\eta, S} \in \GU \, \big\vert\, \text{Ad}_{U_{\eta, S}}(\CS_{\sqr}) = \CS_{\sqr} \Big\} \, .
\end{equation}
One can see that $U_{\eta, S} \in N(\CS_{\sqr})$ requires $S \in \text{Aut}^{\text{S}}(\mathcal{L})$ as otherwise the underlying lattice does not stay invariant.

However, Gaussian unitaries $U_{0, S}$ with $S\in \text{Aut}^{\text{S}}(\mathcal{L})$ do not necessarily leave the sector invariant and generally need to be followed by a displacement in order to constitute an automorphism of the stabilizer group, an issue also discussed in Refs.~\cite{ Conrad_2022_lattice_perspective, Royer_2022_encoding}.
We find
\begin{equation}
    N(\CS_{\sqr}) = \Big\{ U_{\eta, S} \, \big\vert\, \eta \in \CLp + \eta_{S}, \, S \in \text{Aut}^{S}(\CL) \Big\}
\end{equation}
where we have defined $\eta_S = \frac{1}{2\pi} \sum_{i=1}^{2n} \Delta \phi(\xi_i) \xi^{\perp}_{i}$ with $\Delta \phi (\xi) \coloneqq \phi(S^{-1}\xi) - \phi(\xi)$.
The vector $\eta_S$ assumes values within $\frac{1}{2}\CL^{\perp}$  and generally vanishes in the case of symplectically even lattices which include the standard square and hexagonal encodings of a qubit into a single mode. 
Obtaining the normalizer of a code with nontrivial sector is straightforward as each such code differs from the trivial sector by a displacement $D(\xi)$ and the normalizer is then obtained from the trivial sector normalizer as $D(\xi)N(\CS_{\sqr}) D(\xi)^{\dagger}$.

\begin{example}[Qutrit \texttt{S}-Gate]
\label{example:square_GKP_qutrit_Sgate}
Suppose we encode a qutrit into a single mode via the trivial sector of the square GKP code, i.e.\ the code with stabilizer $\CS = \langle D(\xi_1), D(\xi_{2}) \rangle$ where $\xi_{1} = \sqrt{3} (1,0)^{T}$ and $\xi_{2} = \sqrt{3} (0,1)^{T}$. 
If we wish to implement the Clifford operation $\overline{X}_{1}  \mapsto \overline{X}_{1} \overline{Z}_{1}, \, \overline{Z}_{1} \mapsto \overline{Z}_{1}$ we might put to use the lattice automorphism $K$ satisfying $K \xi_{1}^{\perp} = \xi_{1}^{\perp} + \xi_{2}^{\perp}$ and $K \xi_{2}^{\perp} = \xi_{2}^{\perp}$. 
Under this transformation the stabilizer transforms into $\CS' = \hat{K} \CS \hat{K}^{\dagger} = \langle D(\xi_1), D(\xi_{1} + \xi_{2}) \rangle$. 
Note that this is not the original stabilizer as $\CS'$ includes the element $D(-\xi_{1}) D(\xi_1 + \xi_{2}) = - D(\xi_{2})$. 
In fact $\CS' = \langle D(\xi_1), -D(\xi_{2}) \rangle$, differing from the original stabilizer by a sign in front of the second generator. 
In order to remove the  sign from the generator and return to the original code we must apply the follow-up shift operation $T = D(\frac{1}{2}\xi_{2}^{\perp})$ which can be seen to satisfy $T \CS' T^{\dagger} = \CS.$
\end{example}
\vspace{-0.7cm}
\subsubsection{Logical action}
Given $U_{\eta,S} \in N(\CS)$ it is straightforward to obtain its logical Clifford action. 
From $\text{Aut}^{S}(\CL) = \text{Aut}^{S}(\CL^{\perp})$ it follows that the symplectic action $S$ can be rewritten as a unimodular left-action on the dual generator: $M^{\perp}S^{T} = VM^{\perp}$. 
Employing the symplectic representation $\boldsymbol{\sigma} \in (\mathbb{Z}_{d_1} \times \dots \times \mathbb{Z}_{d_n})^{\times 2}$ of a logical Pauli operator $P(\boldsymbol{\sigma}) = \Pi_{i} \overline{X}_{i}^{\boldsymbol{\sigma}_{i}} \overline{Z}_{i}^{\boldsymbol{\sigma}_{i+ n}}$ up to phase, we observe that the symplectic representation vector transforms as
$\boldsymbol{\sigma} \mapsto V^{T} \boldsymbol{\sigma} \text{ mod } D $ where ‘mod $D$’ means that the $i$-th and $n+i$-th rows of a vector are to be taken modulo $D_{ii}$. 
The remaining Pauli part of the Clifford action implemented by $U_{\eta, S}$ is determined by the coset in $\CLp/\CL$ to which $\eta - \eta_{S}$ belongs.
\vspace{-0.1cm}
\begin{example}[Qutrit \texttt{CNOT} Gate]
    Consider two qutrits encoded into two modes via a square GKP code on each mode. 
    Suppose we wish to implement the gate $\texttt{CNOT}_{1 \rightarrow 2}$. 
    Under this action $\boldsymbol{\sigma}$ transforms according to 
    \begin{equation*}
        V^{T} = \begin{pmatrix}
1 & 0 & 0 & 0\\
1 & 1 & 0 & 0\\
0 & 0 & 1 & -1\\
0 & 0 & 0 & 1\\
\end{pmatrix} \, .
    \end{equation*}
    since for the square GKP code $M^{\perp} \propto \mathds{1}$ we have $S = V^{T}$ and the logical gate is implemented by the physical gate $\widehat{S} = \widehat{V^{T}}$.
\end{example}

\begin{figure}
\center
\vspace{-0.3cm}
    \includegraphics[width= \textwidth]{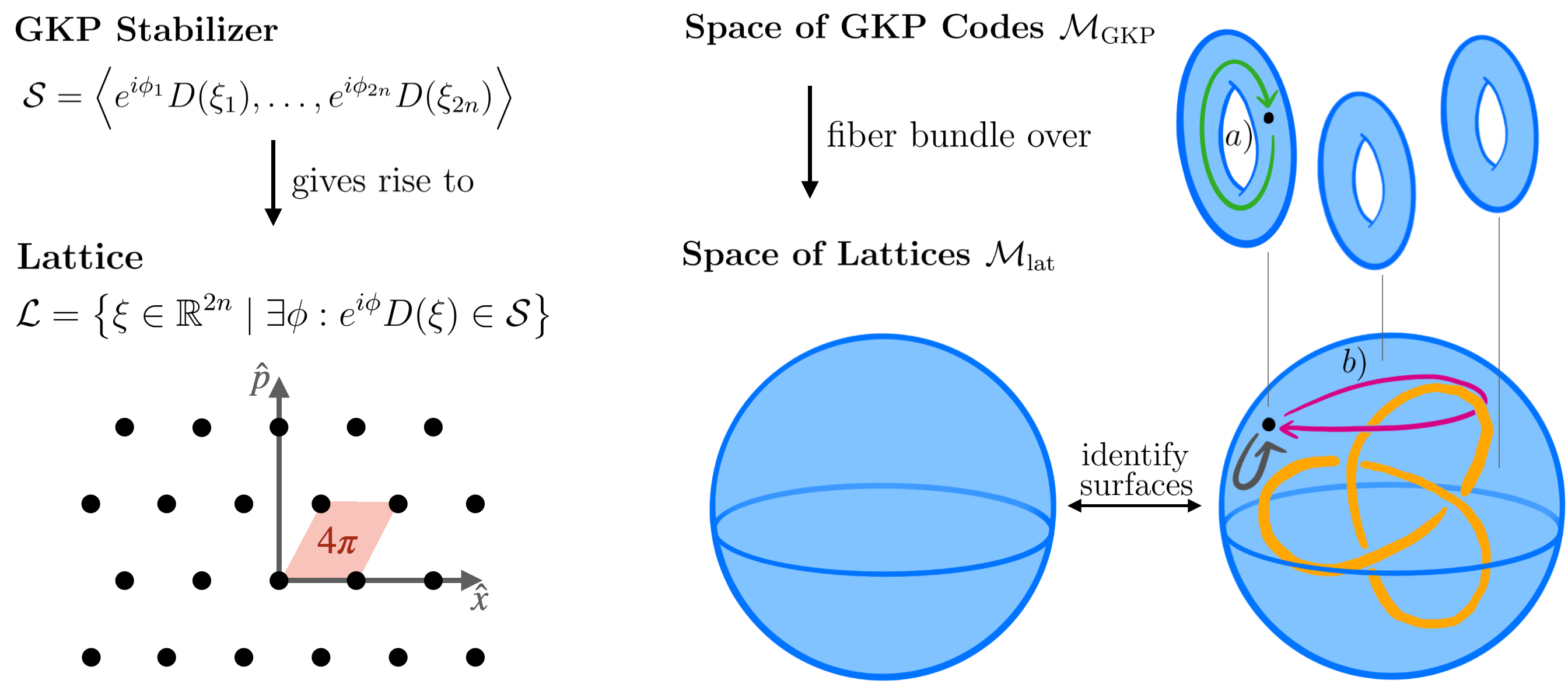}
    \caption{ Illustration of the GKP moduli space. Top Left: A GKP stabilizer is generated by a set of $2n$ displacement operators with phases. Bottom Left: The stabilizer gives rise to a symplectic lattice in phase space. Pictured is the 2-dimensional hexagonal lattice. The unit cell area of this lattice must be a multiple of $2\pi$. The shown area of $4 \pi$ corresponds to an encoded qubit. Top Right: The GKP moduli space is a union of tori, each corresponding to the set of possible phases $\{\phi_{i}\}$ for a given point in the lattice moduli space. a) A path on the GKP moduli space wrapping around a fiber implements a logical Pauli gate. Bottom Right: In the single mode case ($n=1$) each connected component of the space of lattices is topologically equivalent to $S^{3}- \gamma$, the 3-sphere with a trefoil knot removed, the 3-sphere is pictured as two filled balls with surfaces identified. b) Paths on $\mgkp$ whose projections interlink with the trefoil knot can give rise to nontrivial logical Clifford action (up to Paulis). Contractible paths necessarily produce trivial logical Clifford action (up to Paulis). The Pauli part of the logical action is specified by how the path wraps around the fibers over its projection. \vspace{-0.1 cm}}
    \label{fig:moduli_space}
\end{figure}

%% file: section_GKP_moduli_space.tex
Let us now turn our attention to the space of all GKP codes. By definition each GKP code has an underlying lattice $\CL$ and can be identified with its stabilizer
\begin{equation}
    \CS = \Big\langle e^{i \phi_1} D(\xi_1), \dots , e^{i \phi_{2n}} D(\xi_{2n}) \Big\rangle \, .
\end{equation}
Our construction of the GKP moduli space proceeds along the same lines as that of the moduli space of symplectically integral lattices. 
For each lattice type $\CD$, we define a standard GKP code via a stabilizer $\CS_{\sqr}$. 
We then show that each GKP stabilizer $\CS$ with underlying lattice of identical type $\CD$ can be obtained via Gaussian unitary operation as $\CS = U  \CS_{\sqr} U^\dagger$. 
Identifying different Gaussian unitaries producing the same code we obtain the set of all codes with underlying lattice of type $\CD$ as the coset manifold 

\begin{equation}
    \mgkpd = \GU/N(\CS_{\sqr}) \, .
\end{equation}
It follows that the GKP moduli space is given by
\begin{equation}
    \mgkp = \bigcup_{\CD} \mgkpd 
\end{equation}
with the index running over all lattice types as in Theorem~\ref{thm:sympl_integral_lattice_moduli_space}.

\subsection{Standard code}
Given a type $\CD$ we define the \emph{standard code} with respect to the standard lattice generator $M_{\sqr} = (\xi_1, \dots, \xi_{2n})^{T}$ introduced in Corollary~\ref{cor:generator_standard_form}. Specifically, the standard code is the corresponding trivial sector stabilizer
\begin{equation}
    \CS_{\sqr} = \Big\langle D(\xi_1), \dots, D(\xi_{2n}) \Big\rangle .
\end{equation}
Suppose now $\CS$ is another GKP stabilizer with underlying lattice of type $\CD$. 
Then it possesses a generator $N = (\eta_1, \dots , \eta_{2n})^{T}$ bringing the associated Gram matrix into Frobenius standard form. 
By Corollary~\ref{cor:generator_standard_form} there exists a symplectic map $S$ such that $N = M_{\sqr} S^{T}$. Denoting by $\hat{S}$ the metaplectic representation of $S$ we find that $\CS = U \CS_{\sqr} U^\dagger$ with
\begin{equation}
\label{eq:transitive_unitary}
    U = D\bigg(\sum_{i=1}^{2n} \frac{-\phi_{i}}{2\pi} \eta^{\perp}_{i}\bigg)\hat{S}
\end{equation}
where $\phi$ denotes the sector of $\CS$. 
Eq.~\eqref{eq:transitive_unitary} defines a Gaussian unitary matrix.

\subsection{Restriction bundles}
The GKP moduli space is itself a fiber bundle over the lattice moduli space in a natural way. 
Each code projects to its underlying lattice via a map $\pi: \mgkp \rightarrow \mlat$.
The fiber over a given lattice $\CL$ is then given by the set of displacement operators up to displacements in $\CLp$, that is the 2n-dimensional torus $\CT_{\CLp} = \R^{2n} / \CLp$. For a given GKP code in the fiber this torus also constitutes the set of error sectors, or equivalently, their syndromes~\cite{Conrad_2022_lattice_perspective}.
We will refer to this construction as the \emph{GKP bundle}, which is illustrated in Fig.~\ref{fig:moduli_space}. We show below that the GKP bundle is nontrivial as a bundle.
In particular we will show that the restriction of the GKP bundle to an arbitrary loop $\gamma$ in $\mlat$ assumes the structure of a mapping torus. 
\begin{definition}
    The mapping torus of a self-homeomorphism $f: F \rightarrow F$ is the quotient space
    \begin{equation}
       \T_{f} = \frac{ [0,1] \times F }{ \big(0, x\big) \sim \big(1, f(x)\big) } \, .
    \end{equation}
    Equipped with the natural projection map from the unit interval to $S^{1}$ it is a bundle over the circle.
\end{definition}
The ‘twist’ in the mapping torus is determined by the lattice automorphism performed by traversing the loop $\gamma: S^{1} \rightarrow \mlat$. 
In fact, for the lattice automorphism $S \in \text{Aut}^{S}(\CL)$ it will be the inverse of the homeomorphism on the phase space torus $f_{S}: \CT_{\CLp} \rightarrow \CT_{\CLp}$ defined as $f_{S}(\xi) = S \xi \text{ mod } \CLp.$

\begin{figure}
    \centering
\[\begin{tikzcd}
	{p^{*}\mathcal{M}_{GKP}} &&&&& {I \times\mathcal{T}_{\mathcal{L}_{0}^{\perp}}} \\
	& I &&& I \\
	\\
	& {S^{1}} &&& {S^{1}} \\
	{\gamma^{*}\mathcal{M}_{GKP}} &&&&& {\mathbb{T}_{f}}
	\arrow["\varphi"{description}, curve={height=-12pt}, from=1-1, to=1-6]
	\arrow[from=1-1, to=2-2]
	\arrow["{\tilde{p}}"{description}, from=1-1, to=5-1]
	\arrow["{\varphi^{-1}}"{description}, curve={height=-12pt}, from=1-6, to=1-1]
	\arrow[from=1-6, to=2-5]
	\arrow["{q_{f}}"{description}, from=1-6, to=5-6]
	\arrow[from=2-2, to=2-5]
	\arrow["{p'}", from=2-2, to=4-2]
	\arrow["Id", from=2-5, to=2-2]
	\arrow["{p'}", from=2-5, to=4-5]
	\arrow["Id", from=4-2, to=4-5]
	\arrow[from=4-5, to=4-2]
	\arrow[from=5-1, to=4-2]
	\arrow["\theta"{description}, curve={height=-12pt}, from=5-1, to=5-6]
	\arrow[from=5-6, to=4-5]
	\arrow["{\theta^{-1}}"{description}, curve={height=-12pt}, from=5-6, to=5-1]
\end{tikzcd}\]
\caption{Commutative diagram illustrating the proof of Theorem~\ref{thm:restriction_bundles}. Objects on the outer commutative square are fiber bundles, their projection maps pointing to the base spaces on the inner commutative square. We first construct the isomorphism $\varphi$ between two bundles over the unit interval and then construct the sought-after isomorphism $\theta$ through the universal product of the quotient topology.}
\label{fig:comm_diagram}
\end{figure}
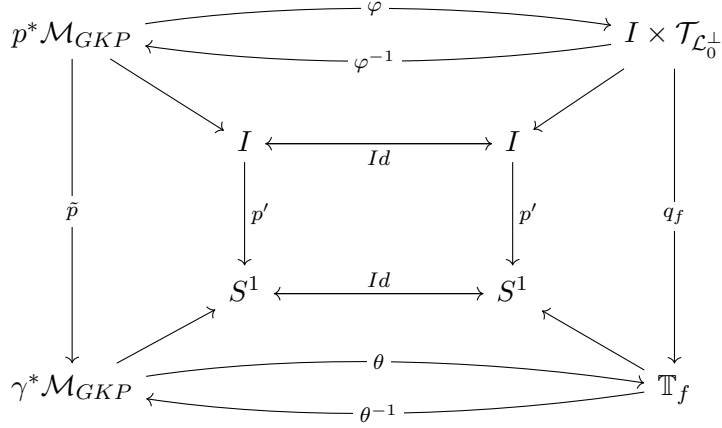

\begin{theorem}
\label{thm:restriction_bundles}
    The pullback bundle $\gamma^{*}\mgkp$ of the GKP bundle along a closed loop $\gamma: S^{1} \rightarrow \mlat$ is isomorphic to a mapping torus $\T_{f}$. The twist $f$ is determined by the lattice automorphism performed by traversing $\gamma$.
\end{theorem}

\begin{proof}
    Let $I = [0,1]$ denote the unit interval and let $p = \gamma \circ p'$ where $p':I \rightarrow S^{1}$ is the natural projection onto the circle.  
    The pullback bundle $p^{*}\mgkp \subset I \times \mgkp$ is a fiber bundle equipped with a bundle morphism $\Tilde{p}\colon p^{*}\mgkp \rightarrow \gamma^{*}\mgkp$.
    As a fiber bundle over the unit interval $p^{*}\mgkp$ is isomorphic to a trivial bundle.
    Let us construct such an isomorphism. First we denote $\CL_{0} = p(0)$ and let $\{\xi_{i} \}$ be a basis of $\CL_{0}$.
    By construction each element of $\mgkp$ is a stabilizer group and we begin by defining the homeomorphism $\overline{\varphi}\colon p^{*}\mgkp|_{0} \rightarrow \CT_{\CLp_{0}}$ which maps the stabilizer $\langle e^{i\phi_1} D(\xi_1), \dots, e^{i\phi_{2n}} D(\xi_{2n})\rangle$ to $\frac{-1}{2 \pi} \sum_{i}\phi_{i}\xi_{i}^{\perp} \text{ mod } \CLp_{0}$. 
    We further denote by $K$ the unique lift of $p$ into $\Sp_{2n}(\R)$ based at the identity \footnote{By a lift of $p\colon I\to \mlat$ into $Sp_{2n}(\mathbb{R})$ we mean a map $K\colon I \to Sp_{2n}(\mathbb{R})$ such that $\forall x\colon K(x) p(0) = p(x)$. Such a lift is not unique as even for $x=0$, where one has $p(x=0) = \mathcal{L}_0$, any $K(0)$ that is an element of $\text{Aut}^{S}(\mathcal{L}_0)$ satisfies the relationship. Among these we make the choice $K(0) = \mathds{1} \in Sp_{2n}(\mathbb{R})$, i.e.\ we base the lift at the identity, which uniquely specifies $K(x)$ for all other $x$.}.
    This allows us to extend $\overline{\varphi}$ to the bundle isomorphism $\varphi: p^{*}\mgkp \rightarrow I \times \CT_{\CLp_{0}}$ with $\varphi(x, \CS) = \big(x, \overline{\varphi}(\hat{K}^{\dagger}(x) \CS \hat{K}(x))\big)$.
    
    Consider now the homeomorphism $f: \CT_{\CLp_{0}} \rightarrow \CT_{\CLp_{0}}$ defined as $f = \pi'' \circ \varphi \circ \sigma \circ \varphi^{-1} \circ \pi'$ where $\pi''$ denotes the projection map of $I \times \CT_{\CLp_{0}}$ onto $\CT_{\CLp_{0}}$, $\pi'(c) = (0, c)$ embeds $\CT_{\CLp_{0}}$ into $I \times \CT_{\CLp_{0}}$ and $\sigma$ is the shift $(0, x) \mapsto (1, x)$ within $I \times \mgkp$.
    Let $\T_{f}$ denote the corresponding mapping torus. 
    Then we have the bundle morphism $q_f: I \times \CT_{\CLp_{0}} \rightarrow \mathbb{T}_{f}$ defined by $q_{f}(x,y) = (p'(x), y )$.
    By the universal property of the quotient topology the composition $q_f \circ \varphi $ descends to a bundle morphism $\theta: \gamma^{*}\mgkp \rightarrow \mathbb{T}_{f}$ since for $x, y \in p^{*}\mgkp$ one has $\Tilde{p}(x) = \Tilde{p}(y) \implies q_f \circ \varphi(x) = q_f \circ \varphi(y)$, and the topology on $\gamma^{*}\mgkp$ agrees with the quotient topology inherited from $p^{*}\mgkp$.
    Conversely we can make use of the property that for $x,y \in I \times \CT_{\CLp_{0}}$ it follows that $q_f(x) = q_f(y) \implies \Tilde{p} \circ \varphi^{-1}(x) = \Tilde{p} \circ \varphi^{-1}(y)$ and the composition $\Tilde{p} \circ \varphi^{-1}$ descends to a continuous inverse of $\theta$ which is thus a bundle isomorphism between $\gamma^{*}\mgkp$ and $\T_{f}$.
    Plugging into the definition of $f$ we see that for $\xi \in \CT_{\CLp_0}$ we have $f(\xi) = K(1)^{-1}\xi + \eta_{K^{-1}}  \text{ mod } \CLp_0$.
    As $f$ is continuously deformable through homeomorphisms into $f_{K^{-1}}(\xi) = K(1)^{-1}\xi  \text{ mod } \CLp_0$, i.e.\ the inverse of the Clifford homeomorphism on $\CT_{\CLp_{0}}$ inherited from $K(1)$, it follows that  $\gamma^{*}\mgkp$ is bundle isomorphic to  $\T_{f_{K^{-1}}}$.
\end{proof}
In particular, this shows that the GKP bundle is nontrivial as restrictions of trivial bundles produce trivial mapping tori only.

%% file: section_geometric_FT.tex
Our aim in this section is to prove the geometric fault tolerance conjecture given by Gottesman and Zhang in Ref.~\cite{gottesman_fibre_2017} for the case of GKP Pauli and Clifford operations. The general version of the conjecture reads as follows. 

\begin{conjecture}[Gottesman and Zhang]
    Fault-tolerant logical gates can always be expressed as arising from monodromies of an appropriate fiber bundle with a flat projective connection.
\end{conjecture}

We prove that the set of Gaussian unitary operators $\GU$ induces a flat projective connection on the trivial vector bundle over the GKP moduli space $\mgkp$ whose fiber over any point $\CC \in \mgkp$ is the space $\R^{2n}$ of displacement operators.
The parallel transport maps of the connection take values within $\GU$ and act on elements of the fiber via adjoint action $D(\xi) \mapsto U(t) D(\xi) U^{\dagger}(t)$. As logical Pauli operators are displacements and hence elements of the fiber, the logical action on codewords induced by parallel transport along any path is fully determined. 
We split the proof into two separate Lemmas below.

\subsection{Flatness of parallel transport}
We begin by showing that the parallel transport induced by Gaussian unitary operations on $\mgkp$ is well defined.

\begin{lemma}
\label{lem:unique_GU_lifts}
    Any path $\CC\colon[0,1] \rightarrow \mgkp$ of GKP codes has a unique Gaussian unitary implementation $U\colon [0,1] \rightarrow \GU$ that satisfies $U(0) = \mathds{1}$ and implements the path via $\text{Ad}_{U(t)} \CC(0) = \CC(t).$
\end{lemma}
\begin{proof}
    As a continuous map $\CC$ only maps into a connected component $\mgkpd = \GU/N(\CS_{\sqr})$. 
    Since $N(\CS_{\sqr})$ is discrete, $\GU$ is a covering space of $\mgkpd$ with projection map $\pi (U) = U \CS_{\sqr} U^{\dagger}$.
    Let $\gamma$ be a lift of $\CC$ into the cover and consider $U(t) = \gamma(t) \gamma(0)^{-1}$. 
    Then $\text{Ad}_{U(t)} \CC(0) = \text{Ad}_{\gamma(t)} \CS_{\sqr} = \CC(t)$ and $U(0) = \mathds{1}$.
    Given another $U'$ with these properties it follows that $\gamma'(t) = U'(t)\gamma(0)$ is a lift of $\CC$ with $\gamma'(0) = \gamma(0)$.
    A general theorem about lifts to covering spaces states that any two lifts $\gamma, \gamma'$ of a path which satisfy $\gamma(0) = \gamma'(0)$ are identical and thus $U(t) = U'(t)$, i.e.\ $U$ is unique. 
\end{proof}

We proceed by establishing the flatness of the connection, meaning that parallel transport maps along any two homotopic paths are necessarily equal.

\begin{lemma}[Flatness]
    \label{lem:flatness}
    Let $\CC_{1}, \CC_{2}\colon [0,1] \rightarrow \mgkp$ be a pair of homotopic paths, then the implementations $U_1, U_2$ provided by Lemma~\ref{lem:unique_GU_lifts} satisfy $U_{1}(1) = U_{2}(1)$.
\end{lemma}

\begin{proof}
     Let $\CC$ be a homotopy between $\CC_{1}$ and $\CC_{2}$, i.e.\ a continuous map  $\CC: [0,1]^{2} \rightarrow \mgkp$ with $\CC(t,0) = \CC_{1}(t)$ and $\CC(t,1) = \CC_{2}(t)$. 
     As homotopic continuous maps $\CC_{1}, \CC_{2}$ only map into a single connected component $\mgkpd = \GU/N(\CS_{\sqr})$, which is covered by $\GU$. 
     Then, as a consequence of the homotopy lifting property of covering spaces, there exists a lifted homotopy $\gamma(t,x)$ between two lifts of $\CC_{1}$ and $\CC_{2}$ into $\GU$. 
     By the arguments within the proof of Lemma~\ref{lem:unique_GU_lifts} the implementing unitaries are given by $U_{1}(t) = \gamma(t,0)\gamma(0,0)^{-1}$ and $U_{2}(t) = \gamma(t,1)\gamma(0,1)^{-1}$. As $\gamma(1,x)$ and $\gamma(0,x)$ are constant as functions of $x$ it follows that $U_{1}(1) = U_{2}(1)$.
\end{proof}

Let us note that any Gaussian unitary operation transforming one GKP code $\CC_1$ into another $\CC_2$ also transforms the set of displacement Pauli operators of the first to those of the second and thus fully determines the transformation of logical information.  Together Lemmas~\ref{lem:unique_GU_lifts} and \ref{lem:flatness} constitute a proof of Theorem~\ref{thm:geometric_ft} and hence the geometric fault tolerance conjecture for the case of GKP Pauli and Clifford operations. 
To the best of our knowledge, this constitutes the first explicit example of geometric fault tolerance in a multi-mode bosonic error correction code; see, however, also our separate paper Ref.~\cite{conrad_2024_gkp_rosetta} that treats the single-mode case using algebraic tools. The main difference between our construction and those used to prove the conjecture for the case of transversal gates and a set of topological gates in Ref.~\cite{gottesman_fibre_2017} is that we do not explicitly construct a connection on a vector bundle over a submanifold within the Grassmanian $\text{Gr}_{K}(\CH)$, the manifold of all dimension-$K$ subspaces of a Hilbert space. 
Instead, the manifold upon which we construct a connection is one consisting of GKP stabilizer groups. The difference stems from the fact that GKP code-states are formally of infinite energy and hence do not form a well-defined subspace, necessitating an indirect description through the stabilizer group. 

While we have exclusively dealt with the idealized infinite energy versions of GKP codes, we expect our main results, Lemmas~\ref{lem:flatness} and \ref{lem:unique_GU_lifts}, to carry over to the finite energy setting. In the finite energy setting the GKP manifold $\mgkp$ needs to be replaced by a  manifold of approximation dependent finite-energy codes. Additionally, while we have focused on gates implemented by Gaussian unitary operations, such gates can in general deform the envelopes of finite energy states, and would need to be replaced by explicitly envelope preserving operations~\cite{Royer_2020_stabilization, matsos_2024_Universal}. We leave the detailed development of such a formulation to future work.

%% file: section_FiberBundleIntroduction.tex
Here we give a short and non-rigorous introduction to the concept of vector bundles and, more generally, fiber bundles. The introduction is intended for readers with little or no prior exposure to bundle theory and omits some technicalities. We refer the interested reader to Refs.~\cite{baez_1994_Gauge, frankel_2011_Geometry} for more detailed treatments.

One of the most commonly occurring objects in physics are functions $f\colon \mathcal{B} \rightarrow \C$ on a manifold $\mathcal{B}.$ Such a function assigns to every point $x \in \mathcal{B}$ a value within the vector space $\C$, or, more generally, within some other vector space $V$. In certain settings, however, it turns out to be more natural to assign to $x \in \mathcal{B}$ not values within one fixed vector space, but within different vector spaces $V_{x}$, one for each point $x\in \mathcal{B}$. This assignment of vector spaces, one to each point, gives rise to the notion of a vector bundle.

\begin{definition}[Vector Bundle]
    A vector bundle is a pair of topological spaces $\mathcal{B}$ and $\mathcal{T}$ together with an onto, continuous map $p\colon\mathcal{T}\rightarrow\mathcal{B}$, such that the pre-image $p^{-1}(x)$ of any $ x \in \mathcal{B}$, has the structure of a vector space.
\end{definition}
 One usually refers to the space $\mathcal{T}$ as the \textit{total space} and to $\mathcal{B}$ as the \textit{base space}. The pre-image $p^{-1}(x)$ of $x \in \mathcal{B}$ is called the \textit{fiber over $x$}. The generalized notion of a function, which assigns to each $x \in \mathcal{B}$ a value within the fiber over it is called a \textit{section}.
 
Simple examples of vector bundles are the cartesian products $\mathcal{B} \times V$ of the base space with a vector space, together with the natural projection $(b, v) \mapsto b$. Such bundles are called \textit{trivial bundles}. Another natural example of a vector bundle is the tangent bundle of a manifold, where the fiber over each point in the base space is simply the space of vectors tangent to that point of the manifold. The main example of a vector bundle discussed within the main text is the case where $\mathcal{B}$ parametrizes a family of error-correcting codes and where the fiber over each point is the the space $\R^{2n}$ of displacement operators. 

One of the most immediate consequences of the vector bundle construction is that one cannot immediately compare the values of a section at different points within the base space, as they assume values within different vector spaces. In order to compare two vectors, they must lie within the same vector space, and one introduces the notion of a \textit{connection} to enable this. A connection mathematically describes the action of transporting a vector along the fibers over a path $\gamma\colon[0,1]\rightarrow \mathcal{B}$. We give the following definition, which, though weaker than the usual definition of a connection, is sufficient for the purpose of this introduction.
\begin{definition}[Connection on a Vector Bundle]
    A connection on a vector bundle $(\mathcal{T}, \mathcal{B}, p)$ assigns to each path $\gamma\colon[0,1]\rightarrow \mathcal{B}$ an invertible linear map $D(\gamma)$ from $ p^{-1}\big(\gamma(0)\big)$ to $ p^{-1}\big(\gamma(1)\big)$. The assignment satisfies $D(\gamma_1 * \gamma_2) = D(\gamma_1) D(\gamma_2)$, where $*$ denotes the concatenation of paths, as well as $D(\gamma^{-1}) = D(\gamma)^{-1}$, and is independent of the parametrization of $\gamma$ \footnote{Reparametrization is any map $\gamma \mapsto\gamma \circ g$, where $g:[0,1]\rightarrow[0,1]$ is continuous with $g(0) = 0,\, g(1) = 1$ and $a \geq b \implies g(a) \geq g(b)$. Reparametrization independence ensures that the transport map depends only on the path, and not the chosen parametrization.}.
\end{definition}
Importantly, the transport map between two points is dependent on the chosen path connecting the two points. A natural example of this path-dependence is the parallel transport of a tangent vector on any curved surface. The parallel transport of a tangent vector from the north pole of a sphere to a point on the equator, for example, depends on the chosen path~\cite{baez_1994_Gauge}.

Connections which yield parallel transport maps that depend only on the topology of the path connecting two points within $\mathcal{B}$ are called \textit{flat}.
\begin{definition}[Flatness of a connection]
    A connection is called flat if the transport maps it assigns to any two homotopic paths are equal.
\end{definition}
In particular flatness implies that the transport map assigned to a contractible loop must be the identity, $\mathds{1}$, on the fiber over the loop's base point. Loops that wrap around a hole in the manifold and are not contractible, however, can give rise to nontrivial transport maps from a fiber to itself.

In the above we have focused, for simplicity, on the concept of vector bundles. More generally, however, one might want to allow target spaces other than vector spaces. Dropping the requirement that the fibers of a bundle be vector spaces then leads to the notion of a \textit{fiber bundle}. While there is no general requirement for all fibers of a bundle to be copies of some reference fiber, this is the case throughout this work. In this case case one refers to the reference fiber as the \textit{standard fiber}.

%% file: section_Frobenius_form.tex
The fact that the Gram matrix allows a unique normal form is a consequence of the following theorem. 
The theorem provided here is an English language translation of the version in Ref.~\cite{Bourbaki_2007_algebre}.

\begin{theorem}[Frobenius~\cite{Frobenius_1879_theorie}]
\label{thm:frobenius}
Let A be a principal ideal domain, E a free A-module of finite dimension $n$ and $\Phi$ an alternating bilinear form on $E$. Then there exists a basis $(e_i)_{1 \leq i \leq n}$ of $E$ and an even integer $2r \leq n$, such that 1: $\Phi(e_1, e_2) = \alpha_1, \Phi(e_3, e_4) = \alpha_2, ..., \Phi(e_{2r-1}, e_{2r}) = \alpha_r$ where the $\alpha_i$ are elements $\neq 0$ of $A$, and where $\alpha_i$ divides $\alpha_{i+1}$ for $i = 1,...,r-1$. 2: All other elements $\Phi(e_i, e_j)$ for $i \leq j$ are zero. The ideals A$\alpha_i (i = 1, ..., r)$ are uniquely determined by the preceding conditions. The submodule $E^{0}$ of $E$ orthogonal to $E$ is generated by $e_{2r+1}, ..., e_{n}$. 
\end{theorem}

\FrobeniusStandardForm*
\begin{proof}
    Every finite dimensional lattice is a free module over the principal ideal domain $\mathbb{Z}$, equipped with an alternating bilinear form by restriction of the symplectic form $J_{2n}$ to $\CL$. Theorem~\ref{thm:frobenius} guarantees a basis $\{ \xi_1, \xi_{1+n}, \xi_{2}, \xi_{2+n}, \dots, \xi_{n}, \xi_{2n} \}$ of $\CL$ such that the generator $M = (\xi_{1}, \xi_{2}, \dots, \xi_{2n})^{T}$ satisfies $M J M^{T} = J_{2} \otimes \CD$ with $\CD = \textnormal{diag}(d_1, \dots, d_n)$ and $d_{n} |d_{n-1} | \dots |d_{1}$. Two ideals $\mathbb{Z} d_1, \mathbb{Z} d_2$ are equal exactly if $d_1 = \pm d_2$ and thus each ideal can be identified with the unique non-negative integer generating it. By the change of lattice basis vector $\xi_{2i} \mapsto -\xi_{2i}$ we obtain the transformation $d_{i} \mapsto - d_{i}$ and $ \forall j \neq i: d_{j} \mapsto d_{j} $, i.e.\ all $d_{i}$ can be taken to be positive, guaranteeing uniqueness of $\CD$.
\end{proof}